\documentclass{article}
\usepackage[utf8]{inputenc}
\usepackage[english]{babel}

\usepackage{amsmath}
\usepackage{graphicx}
\graphicspath{ {./Images/} }
\usepackage[colorlinks=true, allcolors=blue]{hyperref}
\usepackage{amssymb}
\usepackage{dsfont}
\usepackage{algorithm}
\usepackage{algpseudocode}
\usepackage{bm}
\usepackage{amsthm}
\usepackage{arxiv}
\usepackage{lipsum}
\usepackage{float}
\usepackage{cite}
\usepackage{natbib}

\setcitestyle{authoryear,open={(},close={)}} 

\newcommand{\N}{\mathbb{N}}

\newcommand{\R}{\mathbb{R}}

\newcommand{\E}{\mathbb{E}}
\newcommand{\PB}{\mathbb{P}}

\newcommand{\1}{\mathds{1}}

\def\1{\mathds{1}}

\newcommand{\FWER}{\textsc{fwer}}
\newcommand{\floor}[1]{\left\lfloor #1 \right\rfloor}

\newlength{\leftstackrelawd}
\newlength{\leftstackrelbwd}
\def\leftstackrel#1#2{\settowidth{\leftstackrelawd}%
{${{}^{#1}}$}\settowidth{\leftstackrelbwd}{$#2$}%
\addtolength{\leftstackrelawd}{-\leftstackrelbwd}%
\leavevmode\ifthenelse{\lengthtest{\leftstackrelawd>0pt}}%
{\kern-.5\leftstackrelawd}{}\mathrel{\mathop{#2}\limits^{#1}}}

\theoremstyle{plain} 
\newtheorem{theorem}{Theorem}[section]
\newtheorem{lemma}[theorem]{Lemma}

\theoremstyle{definition}
\newtheorem{remark}[theorem]{Remark}
\newtheorem{definition}[theorem]{Definition}
\newtheorem{example}[theorem]{Example}
\newtheorem{assumption}[theorem]{Assumption}

\title{Asymptotic online FWER control for dependent test statistics}

\author{Vincent Jankovic\\
	Competence Center for Clinical Trials Bremen\\
	University of Bremen\\
	\texttt{jankovic@uni-bremen.de} \\
	\And
	Lasse Fischer\\
	Competence Center for Clinical Trials Bremen\\
	University of Bremen\\
	\texttt{fischer1@uni-bremen.de} \\
	\AND
	Werner Brannath\\
	Competence Center for Clinical Trials Bremen\\
	University of Bremen\\
	\texttt{brannath@uni-bremen.de}
}

\renewcommand{\headeright}{A Preprint - \today}
\renewcommand{\undertitle}{}
\renewcommand{\shorttitle}{}

\hypersetup{
pdftitle={Asymptotic online FWER control for dependent test statistics},
pdfsubject={q-bio.NC, q-bio.QM},
pdfauthor={Vincent Jankovic, Lasse Fischer, Werner Brannath},
pdfkeywords={online multiple testing, dependent test statistics, familywise error rate},
}

\begin{document}
\maketitle

\begin{abstract}
    In online multiple testing, an a priori unknown number of hypotheses are tested sequentially, i.e., at each time point, a test decision for the current hypothesis has to be made using only the data available so far. Although many powerful test procedures have been developed  for online error control in recent years, most of them are designed solely for independent or at most locally dependent test statistics. In this work, we provide a new framework for deriving online multiple test procedures that ensure asymptotical (with respect to the sample size) control of the familywise error rate, regardless of the dependence structure between test statistics. In this context, we give a few concrete examples of such test procedures and discuss their properties. Furthermore, we conduct a simulation study in which the type I error control of these test procedures is also confirmed for a finite sample size, and a gain in power is indicated.
\end{abstract}


\section{Introduction}

\label{section_introduction}

In times when ever larger amounts of data become more readily available, there is a growing interest in being able to carry out numerous statistical tests as quickly as possible. To this end, the online setting that was introduced by \citet{foster2008alpha}, provides a flexible framework where it is assumed that an unknown and arbitrarily large number of hypotheses can be tested in a sequential manner. In order to address the arising multiplicity problems, many different online test procedures have been developed over the years. While most of them aim for control of the false discovery rate \citep{foster2008alpha,aharoni2014generalized,javanmard2018online,ramdas2018saffron}, there are also situations in which the control of the familywise error rate (FWER) may be more appropriate. 

This can be the case when optimizing machine learning algorithms \citep{feng2022sequential} or in a platform trial, where different treatment arms start at different time points and use a shared control group \citep{robertson2023online}. \citet{tian2021online} developed online procedures specifically for familywise error rate control, e.g., the powerful ADDIS-Spending algorithm. Furthermore, \citet{fischer2022online} derived an online version of the traditional closure principle, and in \citet{fischer2023adaptive} the application and interpretation of ADDIS procedures was facilitated using graphical approaches.

Despite the great progress in the online literature over the recent years, there is still a lack of methods that work for dependent test statistics or $p$-values, as it was for example also recently addressed by \citet{robertson2022online}. However, in reality, it is often way more plausible that there are certain underlying dependencies.

As a concrete motivational example, consider the case of a platform trial (see Fig.~\ref{figure_diagram_platform trial}). Here, dependency issues arise due to the shared control data. Online algorithms for error control in platform trials have been considered, e.g., in \citet{robertson2023online} and \citet{zehetmayer2022online}. 

But there are also other reasons why observations corresponding to different hypotheses could be highly dependent. For example, the observations regarding different hypotheses could originate from the same individual, or a hypothesis that was not rejected at first could be tested again using old and new data. Another application that causes dependencies would be public databases, where different independent research groups have access to the same data and perform their own tests and analyses \citep{robertson2022online}.

\begin{figure}

	\centering

	\includegraphics[scale=0.6]{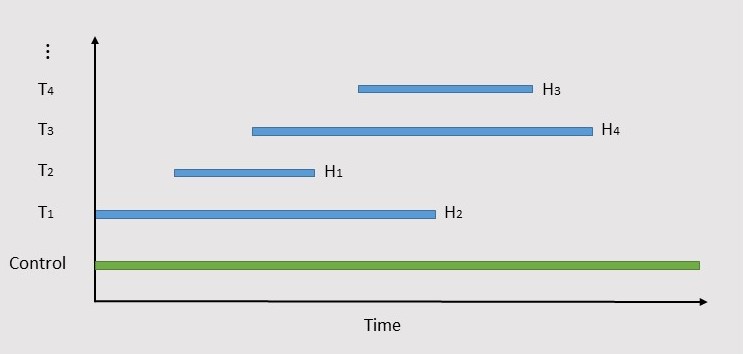}

	\caption{Diagram of a platform trial}\label{figure_diagram_platform trial}

\end{figure}

In the literature so far, the test levels of online procedures depend only on the previous $p$-values or discrete transformations of them  \citep{javanmard2018online,ramdas2018saffron,tian2021online}. When $p$-values are dependent, it is no longer possible to use them directly for specifying test levels. \citet{zrnic2021asynchronous}  circumvented this issue by using only the information of $p$-values that are guaranteed to be independent from the current one, which requires some knowledge of the dependency structure.  In this paper we propose a flexible method that uses the data (instead of just the $p$-value) and can be equally applied, even when there are dependencies in the data and no or few information about the exact dependency structure is available. Thereby we make use of a resampling method that is known as the low intensity  bootstrap (also referred to as $m$ out of $n$ bootstrap) originally proposed in \citet{swanepoel1986note} and thoroughly discussed in \citet{bickel2012resampling}. This is a resampling strategy, where resampling is done with a lower sample size ($m$ instead of $n$) than in the original sample. It can be shown that in some cases this leads to consistency, even when the traditional bootstrap fails \citep{bickel2012resampling}.

The remainder of the paper is structured as follows: In Section\,\ref{section_preliminaries} we introduce the underlying set-up and notation and review some of the already existing methods for online familywise error rate control. The main results are given in Section\,\ref{section_main_results}, where we introduce the notion of consistent weights and, building upon this, derive new online procedures for familywise error rate control under dependency. In Section\,\ref{section_simulations} numerical simulations are performed in two different settings to exemplify the theoretical considerations. The results are again summarized and discussed in Section\,\ref{section_discussion}. The more technical proofs and derivations are given in the \nameref{section_appendix}. 

\section{Preliminaries}\label{section_preliminaries}

\subsection{Notation and model set-up}\label{subsection_notation} Suppose we want to sequentially test null hypotheses $H_1,H_2,H_3,\ldots$ and at each step $i$ we perform the test based on data $\mathcal{X}_i$, generated by the true distribution $\mathbb{P}$. In each case, the corresponding sample size, denoted by $n_i$, is assumed to be data independent.  We make no further assumptions on the data sets, so they possibly could have an arbitrary overlap  for different hypotheses or even coincide. With $n\in\N$ we denote a general parameter of the sample size, i.e., for all $i \in \N$, it holds that $n_i\rightarrow\infty$ as $n\rightarrow\infty$. All asymptotic considerations throughout this paper are performed by letting $n\rightarrow\infty$.  When testing the null hypothesis $H_i$  versus the alternative $H'_i$, the decisions are made in terms of test statistics $T_{i,n}(\mathcal{X}_i)$ or  $p$-values $P_{i,n}(\mathcal{X}_i)$. The latter is more in line with the common assumption in the online multiple testing literature. For ease of notation, we typically just use the index $n$ to denote sample size dependent quantities. For instance, we write $P_{i,n}$ instead of $P_{i,n_i}$.  An \emph{online test procedure} is then a family of test levels $(\alpha_{i,n})_{i\in\N}$ that depend on the previous data, i.e., $\alpha_{i,n}=\alpha_{i,n}(\mathcal{X}_1,\ldots ,\mathcal{X}_{i-1})$. We will not use the more general mathematical definition of an online multiple test suggested in \citet{fischer2022online}.

For an $N\in\N$, let further denote the set of indices corresponding to the true null hypotheses among $H_1,\ldots ,H_N$ by $\mathcal{H}^0=\mathcal{H}^0(N)$ and the number of erroneously rejected hypotheses by $V_n=V_n(N)=\sum_{i\in\mathcal{H}^0}\1\{P_{i,n}\leq\alpha_{i,n}\}$. Then, the familywise error rate is defined as 
\begin{equation*}
	\FWER_n=\FWER_n(N)=\PB(V_n\geq 1).
\end{equation*}
The goal is to derive online test procedures that control the familywise error rate asymptotically, i.e., for any $N\in\N$, it holds that
\begin{equation}\label{def_asymptotic_FWER_control}
	\limsup_{n\rightarrow\infty}\FWER_n(N)\leq\alpha.
\end{equation}
Equation \eqref{def_asymptotic_FWER_control} says that the familywise error rate is controlled asymptotically (with respect to the sample size) for any finite number of hypotheses $N$. In contrast to the case of exact error control, where this would also imply the statement for an infinite number of hypotheses, we can \emph{not} guarantee that 
\[\limsup_{n\rightarrow\infty}\limsup_{N\to\infty}\FWER_n(N)\leq\alpha\]
will hold. For this, further assumptions regarding uniform convergence among all hypotheses might be needed, which are not realistic in this general hypothesis testing set-up. Moreover, \eqref{def_asymptotic_FWER_control} reflects the common situation found in essentially all real-life applications, where the number of hypotheses is a priori unknown but finite.

\subsection{Existing procedures for online familywise error rate control}\label{subsection_existing_methods}

We start reviewing some of the currently existing procedures for online familywise error rate control. The  most straight forward method is the so called \emph{Alpha-Spending} algorithm \citep{foster2008alpha}. The idea is to simply divide the global error budget among all hypotheses, and thus it can be seen as an online version of the weighted Bonferroni correction in traditional multiple testing. This can be implemented by selecting a non-negative sequence $(\alpha_i)_{i\in\N}$ such that 

\[\sum_{i\in\N}\alpha_i=\alpha.\]

It can easily be shown by  the Bonferroni inequality that this procedure controls the familywise error rate regardless of the dependence structure of the $p$-values. Unfortunately, it is generally conservative, so there is an interest in developing new and more powerful online multiple testing procedures.

\citet{tian2021online} derived advanced spending algorithms, first and foremost the powerful \emph{ADDIS-Spending} procedure, which combines the concepts of adaptivity and discarding.  Since we focus in this work on the notion of adaptivity (see also the discussion in Section\,\ref{section_discussion}), i.e. accounting for the proportion of true null hypotheses, we briefly introduce the so called \emph{Adaptive-Spending} \citep{tian2021online}. For this, let $\lambda\in (0,1)$, and we define the so called \emph{candidate indicator} of the $i$-th $p$-value by $C_i=\1\{P_i\leq\lambda\}$. Tian and Ramdas showed that, as long as the $p$-values are independent, it is sufficient for familywise error rate control that the online procedure $(\alpha_i)_{i\in\N}$ satisfies

\begin{equation}\label{eqn_adaptive_spending_general_rule}
	\sum_{i\in\N} \frac{1-C_i}{1-\lambda}\alpha_i\leq\alpha.
\end{equation}

Here, the local test levels $\alpha_i$ are allowed to depend on the indicators $C_1,\dots,C_{i-1}$ but not on $C_i$.

As a concrete example of such a procedure, they proposed to pick a non-negative sequence $(\gamma_i)_{i\in\N}$ such that $\sum_{i\in\N}\gamma_i=1$ and define the test levels as

\begin{equation}\label{def_test_levels_adaptive_spending}
	\alpha_i=\alpha (1-\lambda) \gamma_{t(i)}, 
\end{equation}

where $t(i)=1+\sum_{j<i}(1-C_j)$. 

The idea of \eqref{eqn_adaptive_spending_general_rule} is to decrease the remaining error budget only when a true null hypothesis has been tested. The use of the candidate indicator, as a surrogate for the indicator of a true alternative, relies on the fact that $p$-values corresponding to true alternatives are usually small and therefore  rather unlikely to exceed the threshold $\lambda$. The factor $(1-\lambda)$ accounts for the $p$-values, for which the null hypothesis is true (from now on referred to as \emph{null $p$-values}), that are by chance smaller than $\lambda$. 

Since the test levels in \eqref{eqn_adaptive_spending_general_rule} directly depend on the previous $p$-values, we need the underlying assumption that the null $p$-values are valid conditional to the previous $p$-values, i.e., for all $i\in\mathcal{H}^0$ and $x\in[0,1]$, it holds that

\begin{equation}\label{assumption_conditional_super_uniform}
	\PB (P_i\leq x\mid P_1,\ldots ,P_{i-1})\leq x.
\end{equation}

By using the same proof as in \citet{tian2021online}, it is clear that familywise error rate control is guaranteed when \eqref{eqn_adaptive_spending_general_rule} and \eqref{assumption_conditional_super_uniform} are fulfilled.

While \eqref{assumption_conditional_super_uniform} is always true when the $p$-values are independent, it is easily violated when dependencies between the $p$-values are present, as demonstrated in the following example. 

\begin{example}
	Let $(Z_1,Z_2)$ be bivariate normally distributed with means $\E(Z_1)=\E(Z_2)=0$, variances $\mathrm{var}(Z_1)=\mathrm{var}(Z_2)=1$ and covariance $\mathrm{cov}(Z_1,Z_2)=\rho>0$. Consider the one-sided $p$-values $P_i=1-\Phi(Z_i)$ for $i=1,2$. With $c_1<1/2<c_2$ and $q_i=\Phi^{-1}(1-c_i),i=1,2$, it holds that
	\begin{align*}
		\PB (P_2\leq c_2\mid P_1=c_1)&=\PB (Z_2\geq q_2\mid Z_1=q_1)\\
		&=\PB \left\{(Z_2-\rho q_1)\slash (1-\rho ^2)^{1/2}\geq (q_2-\rho q_1)\slash (1-\rho ^2)^{1/2}\mid Z_1=q_1\right\}\\
		&= 1-\Phi\left\{(q_2-\rho q_1)\slash (1-\rho^2)^{1/2}\right\}\\
		&>1-\Phi(q_2)=c_2.
	\end{align*}
	In the penultimate line it is used that the conditional distribution of $Z_2$ given $Z_1=q_1$ is normal with mean $\rho q_1$ and variance $1-\rho ^2$. Moreover, the final inequality holds, since $q_2<0$, $\rho q_1>0$, and $(1-\rho^2)^{1/2}< 1$ according to the assumptions.

\end{example}

\section{Multiple testing with consistent weights}\label{section_main_results}
\subsection{Consistent weights}
As discussed in Section\,\ref{subsection_existing_methods}, the main idea of the Adaptive-Spending algorithm  is to divide the $p$-values depending on whether they are believed to have arisen under the null hypothesis or the alternative.  This is done by using the information of the indicators $C_i$, which are discrete transformations of the $p$-values, when setting the test levels. Because of that, the $p$-value $P_i$ and the corresponding test level $\alpha_i$ are dependent, unless the $p$-values themselves are independent. In general, this dependency between the $p$-values and the corresponding test level causes the Adaptive-Spending procedures to fail, i.e., error control can no longer be guaranteed. 
To circumvent this issue, we want to generate test levels that use information from the data but become deterministic as the sample size tends to infinity. By this, every $p$-value and their corresponding test level are asymptotically independent. Therefore, we introduce the notion of a weight $\xi_{i,n}=\xi_{i,n}(\mathcal{X}_i)$ that captures, in a sense, the likelihood  of a $p$-value being a null $p$-value and replaces the term $1-C_i$ in Adaptive-Spending algorithms. We now state the formal general property of these weights.

\begin{definition}[Consistent weights]\label{def_consistent_weights}
	Let $\lambda_i\in (0,1)$. We call an array of random variables $(\xi_{i,n})_{i\in\N,n\in\N}$ with values in $[0,1]$ an array of \emph{consistent weights} if there exist constants $c_i\in [0,1],i\in\N$, with $c_i\geq 1-\lambda_i$ for $i\in\mathcal{H}^0$ such that
	\begin{equation}\label{eqn_assumption_xi_consistent}
		\xi_{i,n}\to c_i,\quad (n\to\infty) 
	\end{equation} 
in probability for every $i\in\N$.
\end{definition}

Although not required for family error rate control, the weights $\xi_{i.n}$ should ideally also be close to 0 when the alternative $H_i'$ is true. For example, defining $\xi_{i,n}=1-\lambda_i$ for all $n\in\N$ would obviously satisfy \eqref{eqn_assumption_xi_consistent}, but it would not provide any help in identifying non-null $p$-values. Note, for example, that neither the $p$-values $P_{i,n}$ nor the expressions $1-C_{i,n}=\1\{P_{i,n}>\lambda_i\}$  satisfy the assumption \eqref{eqn_assumption_xi_consistent}, since under the exact null hypothesis $P_{i,n}$ is uniformly distributed on $[0,1]$ for any $n\in\N$.

Next, we illustrate two methods for obtaining weights that satisfy \eqref{eqn_assumption_xi_consistent} from observed data.

\begin{example}\label{example_threshold}
    Assume that in each step $i$, the observed data $\mathcal{X}_i$ consist of random variables $X_{i,1},\ldots ,X_{i,n_i}\sim\mathcal{N}(\mu_i,1)$, and we want to test the hypothesis $H_i:\mu_i \leq 0$ versus the alternative $H_i':\mu_i>0$. The test statistic and the corresponding one-sided $p$-value are  \[Z_{i,n_i}=n_i^{1/2}\overline{X}_{i,n_i},~~~ P_{i,n_i}=1-\Phi (Z_{i,n_i}),\] 
    where $\overline{X}_{i,n_i}$ is the sample mean and $\Phi$ denotes the cumulative distribution function of the standard normal distribution. Furthermore, let $(a_n)_{n\in\mathbb{N}}$ be a  non-decreasing sequence of positive real numbers, so that $a_n\to\infty$ and $n^{-1\slash 2}a_n\to 0$ as $n\to\infty$. For any $\lambda_i\in (0,1)$ we define consistent weights as
    \begin{equation*}
        \xi_{i,n_i}=\begin{cases}
			0&\text{ if }Z_{i,n_i}> a_{n_i}\\
			1-\lambda_i&\text{ if }Z_{i,n_i}\leq a_{n_i}
		\end{cases}.
    \end{equation*}
    It is straightforward to verify that these weights satisfy the conditions in Definition \ref{def_consistent_weights}. Although this method is simple to implement, one drawback is that the sequence $(a_n)_{n\in\N}$ lacks a direct interpretation, making it difficult to choose appropriately. Additionally, it is important to note that these weights are systematically biased in the mean towards 0 under the null hypothesis, which can result in rather anti-conservative procedures for finite sample sizes.
\end{example}

For a more advanced example, we use a resampling technique that is referred to as the low intensity bootstrap. The idea is to resample fewer observations than in the original sample to remedy consistency issues. A commonly given heuristic motivation is that resampling with fewer observations better mimics the relationship of the originally taken sample to the whole ground population \citep{sakov1998using}. 
 The specific resampling scheme will depend on the form of the data $\mathcal{X}_i$ and the test to be conducted. In the following example, which was given in a similar form in \citet{neumann2021estimating}, we illustrate the procedure  for the parametric one-sample test situation as in Example \ref{example_threshold}. The method can be adapted to, e.g., a non-parametric setting \citep{sakov1998using} or 2-sample tests (see \nameref{section_appendix}).

\begin{example}[Low intensity bootstrap example]\label{example_low_intensity}
   We are assuming the same set-up and notation as in Example \ref{example_threshold}. 
Additionally, consider the bootstrap variables $X_{i,1}^*,\ldots ,X_{i,m(n_i)}^*$ resulting from a parametric bootstrap of size $m(n_i)$, i.e., they are independently drawn from the distribution $\mathcal{N}(\overline{X}_{i,n_i},1)$.
	Here, $m:\N\to\N$ is a non-decreasing map such that $m(n)\to\infty$ and $m(n)/n\to 0$ as $n\to\infty$. 
	Additionally, we construct the (low intensity) bootstrap test statistic  $Z_{i,m(n_i)}^*=m(n_i)^{-1/2}\sum_{j=1}^{m(n_i)}X_{i,j}^*$, whose conditional distribution is a normal distribution with mean $\{m(n_i)/n_i\}^{1/2}Z_{i,n_i}$ and variance 1. Moreover, we define the low intensity bootstrap version of the $p$-value $P_{i,n_i}$ by $P_{i,m(n_i)}^*=1-\Phi \{Z^*_{i,m(n_i)}\}$.
	For any $\lambda_i\in (0,1)$, it holds that
	\begin{align*}
		\PB\{P_{i,m(n_i)}^*>\lambda_i\mid\mathcal{X}_i\}&=\PB\{Z_{i,m(n_i)}^*<\Phi^{-1}(1-\lambda_i)\mid\mathcal{X}_i\}\\
		&=\Phi\left[\Phi^{-1}(1-\lambda_i)-\{m(n_i)/n_i\}^{1/2}Z_{i,n_i}\right]\\
		&\to\begin{cases}
			0&\text{ if }\mu_i>0\\
			1&\text{ if }\mu_i<0\\
			1-\lambda_i&\text{ if }\mu_i=0
		\end{cases}.
	\end{align*}
	Thus, by selecting $\xi_{i,n_i}= \PB\{P_{i,m(n_i)}^*>\lambda_i\mid\mathcal{X}_i\}$, property \eqref{eqn_assumption_xi_consistent} is fulfilled.  The low intensity bootstrap, i.e., a bootstrap sample of size $m(n)<n$, is necessary because, when the exact null hypothesis $\mu_i=0$  is true, $Z_{i,n_i}$ is a standard normal random variable regardless of the sample size, and thus, the asymptotic consistency would not be met if $m(n_i)=n_i$.  This issue arises because the bootstrap test statistics are not centralized, as we are not interested in the distribution under the null hypothesis but the true underlying distribution.
\end{example}

As just demonstrated, it is possible to obtain weights $\xi_{i,n}$ from the data by using a low intensity resampling strategy. A common default choice for the size of the resamples is $m(n)=\lfloor n^{1/2}\rfloor$ \citep{neumann2021estimating}. More sophisticated rules regarding the choice of $m$ for the low intensity bootstrap are discussed in \citet{bickel2008choice}.

As mentioned before, besides fulfilling the assumption \eqref{eqn_assumption_xi_consistent} it is also desirable that  $\xi_{i,n}$ tends to 0 under the alternative, which is the case in Example \ref{example_threshold} and Example \ref{example_low_intensity}.

\begin{remark}\label{remark_weights_examples}
    Although both of the mentioned methods lead to consistent weights and therefore are valid choices for the upcoming procedures, we propose to use the bootstrap weights as in Example \ref{example_low_intensity}, because they rely less strictly on pure asymptotics compared to the thresholding weights in Example \ref{example_threshold} and thus seem to perform better in practice. In support of this, it can be shown that for certain situations, procedures using bootstrap weights are guaranteed to have error control not just asymptotically, but also for finite sample sizes (see Theorem~\ref{theorem_finite_sample_size_control}).
\end{remark}

\subsection{Main result}\label{subsection_main_result}

We now present a general framework for incorporating the notions from the previous section into existing algorithms. Therefore, we assume that random variables $\xi_{i,n}$ with  property $\eqref{eqn_assumption_xi_consistent}$ are at hand; for example, they could have been generated according to the low intensity resampling scheme in Example \ref{example_low_intensity}.  

We can now formulate a general rule, very similar to \eqref{eqn_adaptive_spending_general_rule}, for creating online procedures under dependency: Let $\lambda_i\in (0,1)$ for every $i\in\N$. Then,  test levels $(\alpha_{i,n})_{i\in\N}$ have to be chosen in a way such that for every $n,K\in\N$, it holds that
\begin{equation}\label{def_general_dependent_algorithm}
	\sum_{i=1}^K\frac{\alpha_{i,n}}{1-\lambda_i}\xi_{i,n}\leq\alpha.
\end{equation}
Note that, unlike the discrete indicators in \eqref{eqn_adaptive_spending_general_rule}, the weights $\xi_{i,n}$ can take on any value in the range from 0 to 1. Moreover, since we are aiming for asymptotically valid procedures, condition \eqref{def_general_dependent_algorithm} could be relaxed to the effect that it only needs to hold for all but finitely many $n\in\N$, but this has no real practical relevance. In the next two lemmas, we will show that under mild additional assumptions (see Assumptions \ref{ass_consistency_test_levels} and \ref{ass_pvalues_asymptotic_valid}), it is sufficient for asymptotical familywise error rate control to select an online procedure that satisfies \eqref{def_general_dependent_algorithm}.  
\begin{assumption}\label{ass_consistency_test_levels}
The test levels $\alpha_{i,n}=f_i(\xi_{1,n},\ldots ,\xi_{i-1,n})$ are deterministic, continuous functions of the consistent weights. This implies that, due to the definition of consistent weights \eqref{eqn_assumption_xi_consistent} and the continuous mapping theorem, for every $i\geq 1$ it holds that $\alpha_{i,n}\to \alpha_{i}:=f_i(c_1,\dots ,c_{i-1})$ in probability as $n\rightarrow\infty$, where $c_1,\dots ,c_{i-1}$ are the corresponding weight limits.
\end{assumption}

\begin{assumption}\label{ass_pvalues_asymptotic_valid}
 The null $p$-values are asymptotically valid, i.e., for all $i\in\mathcal{H}^0$, there exists a (super-)uniformly distributed random variable $P_i$ such that $P_{i,n}\to P_i$ in distribution as $n\rightarrow \infty$.
\end{assumption}

\begin{remark}\label{remark_consistency_test_levels}
      Assumption \ref{ass_pvalues_asymptotic_valid} is, for example, fulfilled in the test situation described in Example \ref{example_low_intensity}. There, $P_i$ follows a uniform distribution on $[0,1]$ if $\mu_i=0$ and $P_i=1$ almost surely if $\mu_i<0$.
\end{remark}

\begin{lemma}\label{lemma_asymptotical_super_uniformity}
	Let $(\alpha_{i,n})_{i\in\N}$ be an online procedure with asymptotically deterministic test levels, i.e., for every $i\geq 1$, there exists a constant $\alpha_i\in [0,1]$ so that $\alpha_{i,n}\to \alpha_{i}$ as $n\rightarrow\infty$. Under Assumption \ref{ass_pvalues_asymptotic_valid}, it holds that
	\begin{equation*}
		\limsup_{n\rightarrow\infty}\E \left(\sum_{i\in\mathcal{H}^0}\1\{P_{i,n}\leq\alpha_{i,n}\}\right)\leq \sum_{i\in\mathcal{H}^0}\alpha_{i}.
	\end{equation*}
\end{lemma}

\begin{proof}
	Let $i\in\mathcal{H}^0$. By Assumption \ref{ass_pvalues_asymptotic_valid} we have that $P_{i,n}\to P_i$ in distribution as $n\to \infty$, where $P_i$ is a super-uniformly distributed random variable, i.e., $\PB(P_i\leq x)\leq x$ for all $x\in [0,1]$. Due to Slutsky's theorem and the assumption of asymptotically deterministic test levels, we have that $P_{i,n}-\alpha_{i,n}\to P_i-\alpha_i$ in distribution as $n\rightarrow\infty$.
	It follows that
	\begin{align*}
\limsup_{n\rightarrow\infty}\PB(P_{i,n}\leq\alpha_{i,n})&=\limsup_{n\rightarrow\infty}\PB(P_{i,n}-\alpha_{i,n}\leq 0)\\
		&\leq \PB(P_i-\alpha_i\leq 0)\\
		&= \PB(P_i\leq\alpha_i)\\
		&\leq \alpha_i,
	\end{align*}
	where the first inequality holds due to the Portmanteau theorem. The second inequality results from the fact that $P_i$ is (super-)uniformly distributed and $\alpha_i$ is a constant. Due to the linearity of the expectation and the subadditivity of the $\limsup$, the entire inequality also applies to the sum.
\end{proof}

\begin{lemma}\label{lemma_asymptotical_test_levels_conservative}  Let $(\alpha_{i,n})_{i\in\N}$ be an online procedure satisfying \eqref{def_general_dependent_algorithm}. Under Assumption \ref{ass_consistency_test_levels} it holds that 
	\begin{equation*}    
		\sum_{i\in\mathcal{H}^0}\alpha_i\leq\alpha.
	\end{equation*}
\end{lemma}

\begin{proof}
	For all $n\in\N$, we have $ \sum_{i\in \mathcal{H}^0}\frac{\alpha_{i,n}}{1-\lambda_i}\xi_{i,n}\leq \alpha$ almost surely. Due to the definition of  consistent weights \eqref{eqn_assumption_xi_consistent} and Assumption \ref{ass_consistency_test_levels} for every $1\leq i\leq N$, it holds that $\alpha_{i,n}\to\alpha_i$ and $\xi_{i,n}\to c_i$ in probability as $n\to\infty$. Since $\mathcal{H}^0=\mathcal{H}^0(N)$ is finite, it follows that
	\[\sum_{i\in \mathcal{H}^0}\frac{\alpha_{i,n}}{1-\lambda_i}\xi_{i,n}\to\sum_{i\in \mathcal{H}^0}\frac{\alpha_i}{1-\lambda_i}c_i\geq\sum_{i\in\mathcal{H}^0}\alpha_i.\]
	Hence, the lemma.

\end{proof}

Now, we can combine the previous lemmas to obtain the main result.

\begin{theorem}\label{thm_main_error_control}
	Let $(\alpha_{i,n})_{i\in\N}$ be an online procedure  satisfying \eqref{def_general_dependent_algorithm}. Under Assumptions \ref{ass_consistency_test_levels} and \ref{ass_pvalues_asymptotic_valid} it ensures asymptotical control of the familywise error rate, i.e.,
	\[\limsup_{n\rightarrow\infty}  \emph{\FWER}_n\leq\alpha.\]
\end{theorem}

\begin{proof}
	This is a direct conclusion of Lemma \ref{lemma_asymptotical_super_uniformity} and Lemma \ref{lemma_asymptotical_test_levels_conservative} plus the fact that $\PB(V_n\geq 1)\leq\E (V_n)$.	
\end{proof}

We conclude this section with a finite sample size result for the case of independent test statistics. The proof is provided in the \nameref{section_appendix}.

\begin{theorem}\label{theorem_finite_sample_size_control}
    Let $H_1,H_2,H_3,\dots$ be hypotheses concerning the one-sided testing of means and assume that we have (centralized) estimators $\hat{\theta}_{i,n_i}$ for $i\geq 1,$ that are independent and normally distributed with expectation 0 and variance $s_{i,n_i}^2$ under the null hypothesis. The corresponding $p$-values are given by $P_{i,n_i}=1-\Phi (\hat{\theta}_{i,n_i}/s_{i,n_i})$. For each $i\geq 1$, let $\hat{\theta}_{i,m(n_i)}^*$ be bootstrap estimators obtained through a parametric bootstrap procedure using a sample size $m(n_i)$, where $m(n)\to\infty$ and $m(n)\slash n\to 0$ if $n\to\infty$, so that $\hat{\theta}_{i,m(n_i)}^*\mid\mathcal{X}_i~\sim\mathcal{N}(\hat{\theta}_{i,n_i},s_{i,m(n_i)}^2)$. Furthermore, consider the bootstrap $p$-values $P_{i,m(n_i)}^*=1-\Phi \{\hat{\theta}_{i,m(n_i)}^*\slash s_{i,m(n_i)}\}$ and define consistent weights as $\xi_{i,n_i}=\PB \{P_{i,m(n_i)}^*>\lambda_i\mid\mathcal{X}_i\}$, where $\lambda_i\in [0.5,1)$ for all $i\geq 1$. Then, any online procedure that satisfies \eqref{def_general_dependent_algorithm} with respect to the weights $\xi_{i,n_i}$ ensures exact error control, i.e., $\emph{\FWER}_n\leq\alpha$ for every $n\in\mathbb{N}$.
\end{theorem}

\subsection{Test procedures based on consistent weights}\label{subsection_procs_consistent_weights}

In the following, we give a few examples of concrete implementations of  online procedures that assure \eqref{def_general_dependent_algorithm} and thus, according to Theorem \ref{thm_main_error_control}, control the familywise error rate.

\begin{example}\label{example_procedure_uniform}
	For an easy implementation of an online procedure that readily satisfies \eqref{def_general_dependent_algorithm}, we  choose the test levels as
	\begin{equation}\label{eqn_test_levels_canonic_proc}
		\alpha_{i,n}=\Pi_i (1-\lambda_i) \left(\alpha - \sum_{j<i}\frac{\alpha_{j,n}}{1-\lambda_j}\xi_{j,n}\right),
	\end{equation}
	with constants $\Pi_i,\lambda_i\in (0,1)$.
	At any step, the procedure allocates a specific fraction of the remaining alpha budget for the current test.
	This idea, with candidate indicators instead of consistent weights, was originally considered for false discovery rate control in \citet{fisher2021saffron}, where the special case of choosing $\Pi_i=\Pi\in (0,1)$ was proposed and the relation to a geometric spending sequence was discussed. With this selection, the same proportion of remaining alpha wealth would be spent at any step.  This consistent behavior seems reasonable when there is no information on the number of hypotheses that are going to be tested eventually. Another advantage is that there exists a  recursion formula only using the information of the previous test. For $i\geq 1$ it holds that 
	\begin{equation}\label{def_testlevel_uniform_recursion}
		\alpha_{i+1,n}=\Pi_{i+1} (1-\lambda_{i+1}) \frac{\alpha_{i,n}(1-\Pi_i\xi_{i,n})}{(1-\lambda_i)\Pi_i}.
	\end{equation}
	When using a constant $\Pi\in (0,1)$ as in \citet{fisher2021saffron} and with the common default choice $\lambda_i\equiv\lambda\in (0,1)$, the update rule further simplifies to 
	\[\alpha_{i+1,n}=\alpha_{i,n}(1-\Pi\xi_{i,n}).\]
\end{example}

\begin{example}[Continuous Adaptive-Graph]\label{example_procedure_graph}
	For a more advanced example, we first briefly review the ADDIS-Graph that was proposed  in \citet{fischer2023adaptive} for independent or locally dependent $p$-values. This is an online procedure that combines the ideas of adaptivity and discarding \citep{tian2021online} and the Online-Graph \citep{tian2021online,fischer2023adaptive}, where sometimes test levels can be redistributed to future hypotheses.
	To this end, let $(\gamma_i)_{i\in\N}$ and $(g_{j,i})_{i=j+1}^\infty,j\in\N,$ be non-negative sequences that sum to at most 1 and $\tau_i\in(0,1],\lambda_i\in[0,\tau_i)$ for all $i\in\N$.
	The test levels are then defined as
	\begin{equation*}
		\alpha_{i}=(\tau_i - \lambda_i)  \left\{\alpha \gamma_i + \sum_{j=1}^{i-1} g_{j,i}(C_j-S_j+1)\frac{\alpha_{j}}{\tau_j-\lambda_j}\right\}, 
	\end{equation*}
	where $C_j=\1\{P_j\leq\lambda_j\}$ and $S_j=\1\{P_j\leq\tau_j\}$ denote the indicators for a candidate or non-discarding, respectively. The idea is that in certain cases, namely when a $p$-value is either classified as a candidate ($C_j=1$) or discarded ($S_j=0$), the test level can be proportionally distributed  to the future hypotheses according to the weights $g_{j,i}$.
	To adapt this procedure to the new framework, the candidate indicators are substituted by the consistent weights $\xi_{j,n}$. Since in this paper we are solely dealing with the concept of adaptivity, we set $\tau_j=1$ and thus $S_j=1$ for all $j\in\N$, meaning that no $p$-value is discarded. 
	This results in the so called \emph{continuous Adaptive-Graph}, whose test levels are then defined by
	\begin{equation}\label{def_test_levels_dependent_graphgeneral_case}
		\alpha_{i,n}=(1 - \lambda_i)  \left\{\alpha \gamma_i + \sum_{j=1}^{i-1} g_{j,i}(1-\xi_{j,n})\frac{\alpha_{j,n}}{1-\lambda_j}\right\}. 
	\end{equation}
	The proof that the continuous Adaptive-Graph defined in \eqref{def_test_levels_dependent_graphgeneral_case} indeed fulfills condition \eqref{def_general_dependent_algorithm} is given in the \nameref{section_appendix}. 
	It can be shown that the procedure in Example \ref{example_procedure_uniform} with constant $\Pi_i=\Pi\in (0,1)$ is a special case of \eqref{def_test_levels_dependent_graphgeneral_case} by selecting the sequence $\gamma_i=\Pi (1-\Pi)^{i-1}$. For the detailed derivation, we refer to the \nameref{section_appendix}.
\end{example}

\subsection{Continuous spending procedures}\label{subsection_alternative_procs}
In Section\,\ref{subsection_main_result}, we established the general rule \eqref{def_general_dependent_algorithm} that can be used to produce asymptotically valid online test procedures.  While this condition is sufficient, there are also other ways to design asymptotical online test procedures, which do not necessarily fulfill \eqref{def_general_dependent_algorithm}.
To exemplify this, we now provide an example of a class of online test procedures that can be seen as a natural adaptation of the Adaptive-Spending algorithm \eqref{def_test_levels_adaptive_spending}.

For a given $\lambda\in (0,1)$ and a  non-increasing and continuous function $f:[1,\infty)\to \R_{\geq 0}$ with $\int_1^\infty f(x)dx <\infty$, the test levels will be assigned as follows:
\begin{equation}\label{def_test_level_continuous_spending}
	\alpha_{i,n}=\alpha \frac{(1 - \lambda)}{s} f \left(1 + \sum_{j<i}\xi_{j,n}\right),
\end{equation}
where $s=(1-\lambda)f(1)+ \int_1^\infty f(x)dx$ is a constant scaling factor.

This allocation rule states that in each step we lose a bit of wealth determined by the fractions $\xi_{i,n}$. This is similar to \eqref{def_test_levels_adaptive_spending}, but the discrete sequence $(\gamma_i)_{i\in\N}$ and indicator functions $\1\{P_i\leq\lambda_i\}$ are now replaced by a continuous function and the consistent weights, respectively. In this sense, test procedures according to \eqref{def_test_level_continuous_spending} could be referred to as \emph{continuous adaptive spending procedures}.
Note that $\lambda\in (0,1)$ is now assumed to be the same for all hypotheses. It is easy to see that for a continuous adaptive spending procedure, equation \eqref{def_general_dependent_algorithm} does not necessarily hold. Let e.g. $\lambda=1/2$ and $f(x)=2 x^{-4}$. Since $s=(1-\lambda)f(1)+\int_1^\infty f(x)dx=5\slash 3$, the corresponding test levels are  $\alpha_{i,n}=(3\alpha /5) (1 + \sum_{j<i}\xi_{j,n})^{-4}$. It follows that the first summand of the left hand side of \eqref{def_general_dependent_algorithm} is equal to $ (6\alpha/5)\xi_{1,n}$. In general, this expression is not smaller than $\alpha$, as $\xi_{1,n}$ can take on any value between 0 and 1. 

Before we prove error control, we state the following lemma, which takes on a similar role as Lemma \ref{lemma_asymptotical_test_levels_conservative} in Section\,\ref{subsection_main_result}. The proof is given in the \nameref{section_appendix}. 

\begin{lemma}\label{lemma_continuous__spending_test_levels_conservative}
	Let $\lambda\in (0,1)$, $f:[1,\infty)\to \R_{\geq 0}$ be a  non-increasing, continuous, integrable function, and $(\alpha_{i,n})_{i\in\N}$ be an online procedure with test levels according to \eqref{def_test_level_continuous_spending}. Then, under Assumption \ref{ass_consistency_test_levels}, it holds that
	\begin{equation*}    
		\sum_{i\in\mathcal{H}^0}\alpha_i\leq\alpha.
	\end{equation*}
\end{lemma}

Now, we can state the main theorem for formal error control.

\begin{theorem}\label{thm_error_control_continuous_spending}
	Let $(\alpha_{i,n})_{i\in\N}$ be an online procedure according to \eqref{def_test_level_continuous_spending}. Under the same conditions as in Lemma \ref{lemma_continuous__spending_test_levels_conservative} and Assumption \ref{ass_pvalues_asymptotic_valid}, it holds that
	\[\limsup_{n\rightarrow\infty} \emph{\FWER}_n\leq\alpha.\]
\end{theorem}

\begin{proof}
	Since the test levels are continuous transformations of the consistent weights, we can apply Lemma \ref{lemma_asymptotical_super_uniformity} to obtain
	\[\limsup_{n\rightarrow\infty} \FWER_n\leq\sum_{i\in\mathcal{H}^0}\alpha_{i}.\]
	Now, we can use Lemma \ref{lemma_continuous__spending_test_levels_conservative} to complete the proof.
\end{proof}

To conclude this section, we give a concrete example of an online procedure according to the continuous spending approach \eqref{def_test_level_continuous_spending}. 

\begin{example}[Interpolation]\label{example_procedure_interpolation}
	Let $(\gamma_i)_{i\in\N}$ be a non-negative, non-increasing sequence  with $\sum_{i\in\N}\gamma_i =1$. Now we define 
	\begin{equation}\label{eqn_interpolation_function}
		f_\gamma(x)=\gamma_{\floor{x}}+(x-\floor{x}) (\gamma_{\lceil x\rceil} - \gamma_{\floor{x}}),
	\end{equation}
	where $\floor{x}$ and $\lceil x\rceil$ denote the greatest integer smaller\slash least integer greater than or equal to $x$, respectively. 
	The function $f_\gamma$ is the linear interpolation of the sequence $(\gamma_i)_{i\in\N}$, i.e., in particular, it is non-increasing, continuous, and $f_\gamma(x)=\gamma_x$ for all $x\in\N$.
	The test levels according to \eqref{def_test_level_continuous_spending} are defined as
	\begin{equation}\label{def_test_level_interpolation}
		\alpha_{i,n}=\alpha \frac{(1 - \lambda)}{s} f_\gamma \left(1 + \sum_{j<i}\xi_{j,n}\right),
	\end{equation}
	where the scaling factor simplifies to $s=1+\gamma_1 (1/2-\lambda)$ (see \nameref{section_appendix}).
\end{example}

\subsection{Closed testing procedures}\label{subsection_closed_testing}

One of the most important concepts for control of the familywise error rate in the classical offline framework is the \emph{closure principle} by \citet{marcus1976closed}. Here, every intersection hypothesis is tested with a local test, and the resulting multiple test rejects an elementary hypothesis $H_i$ if and only if every intersection hypothesis that is contained in $H_i$ is rejected by the corresponding local test, granting strong control of the familywise error rate. At first glance, this approach does not seem to be suitable for the online setting, since it requires the prior knowledge of all hypotheses at once. However,  \citet{fischer2022online}  developed an online version of the closure principle including a predictability condition, which ensures that in each step only the intersection of hypotheses up to this point needs to be considered. Moreover, they demonstrated a way of converting these online closed procedures into corresponding short-cut procedures.
In this section, we apply  the online closure principle derived in \citet{fischer2022online} to obtain uniform improvements of the online procedures from the previous sections. The central idea is that if any  hypothesis is rejected, the corresponding test level could be reused for future hypotheses. Depending on the procedure, this can, for instance, be implemented by shifting or distributing the regained test level. The Online-Fallback procedure in \citet{tian2021online} can, for example, be seen as a closure of the Alpha-Spending algorithm.  

Let $R_{i,n}=\1\{P_{i,n}\leq\alpha_{i,n}\}$ be the rejection indicator of the $i$-th hypothesis. Then, a closed version of the continuous Adaptive-Graph \eqref{def_test_levels_dependent_graphgeneral_case} is given by
\begin{equation}\label{def_test_levels_dependent_graph_closed}
	\alpha_{i,n}=(1 - \lambda_i)  \left(\alpha \gamma_i + \sum_{j=1}^{i-1} g_{j,i}\max\{1-\xi_{j,n},R_{j,n}\}\frac{\alpha_{j,n}}{1-\lambda_j}\right).
\end{equation}
Similarly, we obtain a closed procedure for the continuous spending approach \eqref{def_test_level_continuous_spending}:
\begin{equation}\label{def_test_level_continuous_spending_closed}
	\alpha_{i,n}=\alpha \frac{(1 - \lambda)}{s} f \left\{1 + \sum_{j<i}(1-R_{j,n})\xi_{j,n}\right\},
\end{equation}
where $s=(1-\lambda)f(1)+ \int_1^\infty f(x)dx$.

It is apparent that both of these closed procedures are uniform improvements of the original ones, respectively.
For the derivation of the formulas, we refer to the \nameref{section_appendix}.

\section{Simulations}\label{section_simulations}
\subsection{General set-up}
To investigate the theoretical results for a finite sample size, we perform a simulation study. Here, we consider two different settings. The first one is a general scenario, where test statistics that are close in time are strongly correlated (e.g., due to data reuse), and those which are further apart from each other have a low or close to zero correlation. The second set-up is more specific and mimics the scenario of a platform trial, where different treatment arms enter the trial at different time points and are tested against a shared control group.

In both settings we explore the closed versions of the proposed continuous Adaptive-Graph \eqref{def_test_levels_dependent_graph_closed} and the continuous adaptive  spending approach \eqref{def_test_level_continuous_spending_closed} (in this section referred to as \emph{Continuous-Graph} and \emph{Continuous-Spending}, respectively), as well as the Online-Fallback procedure \cite{tian2021online}, that grants familywise error rate control regardless of the dependency structure, for comparison. Furthermore, we include the original Adaptive-Spending  \eqref{def_test_levels_adaptive_spending}, for which no online familywise error rate control is guaranteed in these scenarios, to evaluate the impact on the power caused by the adjustments for dependency. In all of the following simulations we choose $g_{j,i}=\gamma_{i-j}$, $f=f_\gamma$ as defined in \eqref{eqn_interpolation_function} and set the remaining hyperparameters as $\lambda=1/2$ and $\gamma_j=6(\pi j)^{-2}$. The consistent weights are defined based on the approach in Example \ref{example_low_intensity}, as they demonstrated better performance in error control compared to the weights in Example \ref{example_threshold}. The power is defined as the expected proportion of truly rejected null hypotheses out of all false null hypotheses. Throughout, we perform 20000 independent simulation runs to estimate the familywise error rate (FWER) and power.

\subsection{Autocorrelation}\label{subsection_simulations_general}

We test in a Gaussian setting $N$ null hypotheses of  the form $H_i:\mu_i \leq 0$ versus the corresponding alternative  \mbox{$H'_i:\mu_i> 0$}. In every run, the parameters $\mu_i$ are set according to the following:

\[\mu_i= \begin{cases}
	0 &\text{ with probability } 1-\pi_1\\
	5 &\text{ with probability } \pi_1
\end{cases},\]

where $\pi_1\in[0,1]$ is the proportion of false null hypotheses. After setting the parameters, we simulate $Z$-Scores according to a multivariate normal distribution, i.e., $(Z_1,\ldots ,Z_N)\sim \mathcal{N}(\mu,\Sigma)$ where $\mu=(\mu_1,\ldots ,\mu_N)$ is the vector of true effects and $\Sigma=(\Sigma_{ij})_{i,j}\in \R^{N\times N}$ is a covariance matrix. We investigate the plausible case where test statistics that are close in time are strongly correlated and the correlation decays rather quickly when test statistics are further apart. Therefore, we set the covariance matrix as  

\[\Sigma_{ij}=\begin{cases}
	1 &\text{ if } i=j\\
	\rho^{|i-j|} &\text{ if } i\neq j
\end{cases}
\]

with $\rho \in \{0.5,0.8\}$.
The  weights $\xi_{i,n}$ required for the Continuous-Graph and the Continuous-Spending are constructed by a low intensity bootstrap as in Example \ref{example_low_intensity} based on the assumed sample size $n$. Due to the parametric set-up, the exact bootstrap probabilities $\PB\{P_{i,m(n)}^*>\lambda\mid\mathcal{X}_i\}$ can be calculated using only the shrunk version of $Z_i$ and the cumulative distribution function of the standard normal distribution, so that no actual resampling is needed.
For the number of hypotheses, we choose $N=1000$, and for the sample size $n=100$, while varying the proportion of true alternatives $\pi_1$. In an additional simulation, we fix $\pi_1=0.1$ and investigate different combinations for the number of hypotheses and the theoretical sample size. Here, we down-scale the effect size with respect to the sample size so that the power remains comparable.   

\begin{figure}
	\centering
	\includegraphics[scale=0.75]{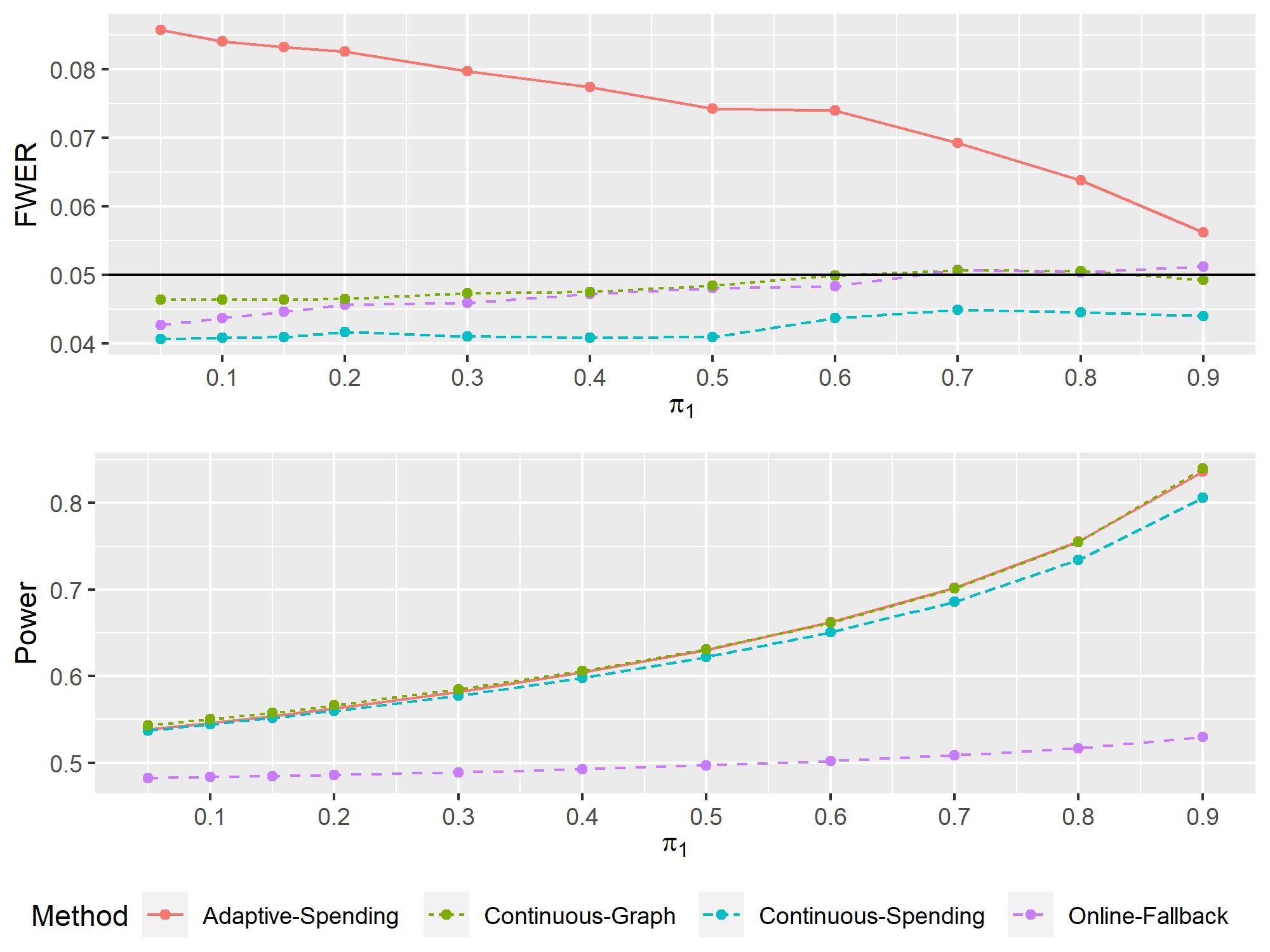}
	\caption{Estimated FWER and power in an autocorrelation set-up with $\rho=0.8$ and target level $\alpha=0.05$ plotted over the proportion of false null hypotheses $\pi_1$.}\label{figure_simulation_autocorrelation_08}
\end{figure}

It can be seen in Fig.~\ref{figure_simulation_autocorrelation_08} that the Continuous-Spending, Continuous-Graph, and Online-Fallback seem to control the familywise error rate at the nominal level $\alpha$. Only the Adaptive-Spending  shows a substantial inflation of 60 \% when the proportion of false null hypotheses is small, which gradually decreases to 20 \% as the proportion of false null hypotheses increases. In terms of power, it is apparent that the Online-Fallback procedure performs considerably worse than the other three procedures, which is no surprise, as it is the only procedure without a proper built-in mechanism for adaption to the number of true null hypotheses. Here, the Continuous-Graph and the Adaptive-Spending are virtually indistinguishable, while both seem to be slightly more powerful than the Continuous-Spending, where the difference becomes clearer for larger values of $\pi_1$. Overall,  the procedures for asymptotic error control seem to work well even with a finite sample size. One illustrative explanation is that the usage of the consistent weights, or more concretely the low intensity bootstrap, in a sense reduces the degree of dependence within the information used to construct the test levels, so that the resulting procedures are comparatively more robust to correlations within the $p$-values.

\subsection{Platform trial}\label{subsection_simulations_platform}

We follow the set-up described in \citet{robertson2023online}. Let $T_1,\ldots ,T_N$ be treatments that enter the platform trial at time points $t_1,\ldots ,t_N$, respectively. Throughout the trial, patients are allocated at a uniform rate to the active treatment arms and the common control arm $C$. 
The responses are given as Gaussian random  variables $X_{ij}\sim\mathcal{N}(\mu_i,\sigma^2)$, $i=0,\ldots ,N$ with known standard deviation $\sigma$. The index 0 corresponds to the control patients. A treatment is generated as effective with probability $\pi_1$. We set $\mu_i=0$ if $T_i$ is ineffective, $\mu_i=0.5$ if it is effective, and $\mu_0=0$. For simplicity, we assume that there will be $n$ patients recruited to each treatment arm $T_i$ when  the according test is eventually performed. Due to the uniform recruitment rate, there will also be $n$ concurrent control patients for each treatment arm. Let $\mathcal{J}_i$ be the index set of the concurrent control observations for the treatment arm $T_i$.
The test statistic for the $i$-th test is of the form 

\[Z_i=\sigma^{-1}(n/2)^{1/2}\left\{\overline{X}_i - \overline{X}^{(i)}_0\right\},\]

where $\overline{X}_i=n^{-1}\sum_{j=1}^nX_{ij}$ and  $\overline{X}^{(i)}_0=n^{-1}\sum_{j\in\mathcal{J}_i}X_{0j}$.
We assume that per unit time there will be 10 patients allocated to all active treatment arms and the control arm. The time point when the treatment $T_i$ enters the trial is defined by $t_i=2(i-1)$. 
Similar to the previous setting, the random variables $\xi_{i,n}$ are calculated via a parametric low intensity bootstrap in the treatment group and the control group, respectively (see \nameref{section_appendix} for the derivation).
We choose $\alpha=0.05$, $\sigma=1$, and $n=100$, i.e., every treatment arm stays in the trial for 10 units of time. For the number of overall treatments, we consider $N\in\{30,50,100\}$.

\begin{figure}
	\centering
	\includegraphics[scale=0.75]{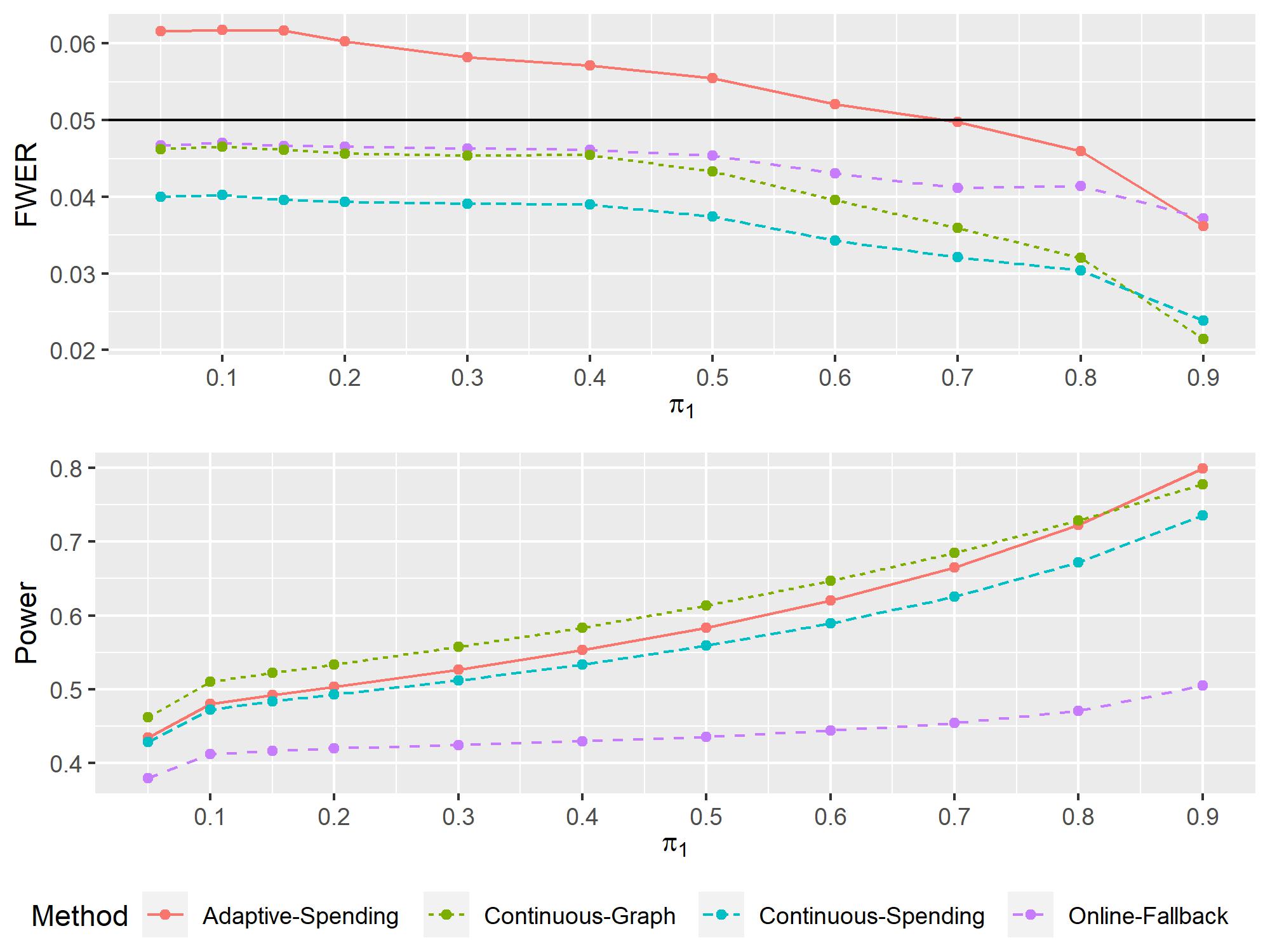}
	\caption{Estimated FWER and power in a platform trial scenario with $N=50$ and target level $\alpha=0.05$ plotted over the proportion of false null hypotheses $\pi_1$.}\label{figure_simulation_platform_N50}
\end{figure}

As can be seen in Fig.~\ref{figure_simulation_platform_N50}, all procedures besides Adaptive-Spending control the familywise error rate. The latter has a moderate error inflation, when $\pi_1$ is smaller than $0.6$. It is noticeable that the procedures (especially  Continuous-Graph and  Continuous-Spending) become quite conservative when the proportion of true alternatives is high. One possible reason for this is, that the number of actually tested hypotheses in this set-up is comparatively small and  a considerable proportion of the regained test level is saved for future hypotheses. This might be counteracted by a different choice of the hyperparameters $(\gamma_i),(g_{j,i})$ and $f$.      
Like in the first setting, we further see that Adaptive-Spending and the Continuous-Graph are the most powerful procedures, after that comes the Continuous-Spending and Online-Fallback is again the least powerful method by far.

\section{Discussion}\label{section_discussion}
Existing advanced online multiple testing procedures only guarantee familywise error rate control under the assumption of either independent or at most locally dependent $p$-values. In practice, independence is an unrealistic assumption, and even methods that work for local dependence can not be applied in some situations, e.g., when the dependence of $p$-values cannot be linked to the intervening time alone or when there is little or no information about the dependency structure at all.    
In this work we proposed a way to derive online procedures that asymptotically control the familywise error rate, regardless of the underlying dependency structure.  To this end, we established general rules for the derivation of new procedures and gave concrete examples. Moreover, we showed how to obtain further improvements of those procedures by using the online closure principle. In a simulation study, we showed that the asymptotic results are also satisfactory for a moderate finite sample size. 

It is worth noting that all proposed procedures, excluding the closed procedures in Section\,\ref{subsection_closed_testing},  effectively control the expected number of false rejections. So they can as well be used to obtain control of the $k$-FWER, i.e., the probability of making at least $k$ false rejections, at target level $\alpha$ by choosing the nominal level for the procedures as $k\alpha$ (see also the discussion in \citet{tian2021online}). It might be worthwhile to investigate whether this new approach of dealing with dependency can be applied to further error rates, as for example the (marginal) false discovery rate.

One current shortcoming is that, although the most powerful procedures developed for the case of independent (or locally dependent) $p$-values use a combination of adaptivity and discarding, only the concept of adaptivity could be transferred to this new setting so far. This has the consequence that the proposed procedures for dependent $p$-values will lose power when the $p$-values are conservative, i.e., being stochastically larger than a standard uniform random variable.
Interesting is that \citet{robertson2023online} did not detect any error inflation of the ADDIS-Spending procedure, which grants formal familywise error rate control only for independent $p$-values, while we saw in Section\,\ref{section_simulations} that the Adaptive-Spending exceeds the nominal test level in a similar simulation set-up. So, it seems that the additional discarding not only improves the power but also somehow helps to maintain the nominal test level in this case. It might be interesting to investigate how this would behave if the $p$-values were negatively dependent. In any case, it would be desirable to incorporate any form of discarding or another method for dealing with conservative $p$-values into the new framework.

Another task for the future would be the extension to more general test set-ups, as  only one-sided 1-sample and 2-sample hypothesis tests have been considered in this work. Similar resampling schemes can probably be used to obtain consistent weights in different situations, e.g., regression models. 

In addition, there is still interest in developing more efficient methods for testing an initially unknown but finite number of hypotheses. Improvements could be expected here if the design of the procedures takes into account that ultimately only finitely many hypotheses are tested.

\clearpage

\section*{Appendix}\label{section_appendix}
\subsection*{Deferred proofs and derivations}

\begin{proof}[Proof of Theorem~\ref{theorem_finite_sample_size_control}]
As in the proof of Theorem~\ref{thm_main_error_control}, we will effectively show that $\E(V_n)\leq\alpha$ and thus, in particular, the familywise error rate is controlled.

By definition, it holds that
\begin{align*}
    \mathbb{E}(\xi_{i,n_i})&=\mathbb{E}\left[\PB \{P_{i,m(n_i)}^*>\lambda_i\mid\mathcal{X}_i\}\right]\\
    &=\PB\{\hat{\theta}_{i,m(n_i)}^*\leq s_{i,m(n_i)} \Phi^{-1}(1-\lambda_i)\}.
\end{align*}

Under the null hypothesis, the unconditional distribution of $\hat{\theta}_{i,m(n_i)}^*$ is a normal distribution with expectation 0 and variance $s_{i,n_i}^2+s_{i,m(n_i)}^2$. 
Since $\lambda_i$ is assumed to be larger or equal than $0.5$ (and thus $\Phi^{-1}(1-\lambda_i)\leq 0$), it holds that $\E (\xi_{i,n_i})\geq 1-\lambda_i$ for all $i\geq 1$. This implies $\mathbb{E}(\frac{\xi_{i,n_i}}{1-\lambda_i}|\mathcal{F}^{i-1})\geq 1$ almost surely, because $\xi_{i,n_i}$ is independent of $\mathcal{F}^{i-1}:=\sigma (\mathcal{X}_1,\dots ,\mathcal{X}_{i-1})$.

It follows that
    \begin{align*}
        \mathbb{E}(V_n)&=\sum_{i\in\mathcal{H}^0}\mathbb{E}\{\mathbb{E}(\1\{P_{i,n_i}\leq \alpha_{i,n}\}|\mathcal{F}^{i-1})\}\\    
        &\leq \sum_{i\in\mathcal{H}^0}\mathbb{E}(\alpha_{i,n})\\
        &\leq \sum_{i\in\mathcal{H}^0}\mathbb{E}\left\{\alpha_{i,n}\mathbb{E}\left(\frac{\xi_{i,n_i}}{1-\lambda_i}|\mathcal{F}^{i-1}\right)\right\}\\
        &= \sum_{i\in\mathcal{H}^0}\mathbb{E}\left(\alpha_{i,n_i}\frac{\xi_{i,n_i}}{1-\lambda_i}\right)\\
        &\leq \sum_{i\in\mathbb{N}}\mathbb{E}\left(\alpha_{i,n_i}\frac{\xi_{i,n_i}}{1-\lambda_i}\right)\\
        &\leq\alpha,
        \end{align*}
where the first inequality holds due to the independence assumption and in the second equality it is used that $\alpha_{i,n}$ is measurable with respect to $\mathcal{F}^{i-1}$.
\end{proof}

\begin{proof}[Proof that the continuous Adaptive-Graph fulfills  condition \eqref{def_general_dependent_algorithm}] 
	Since no asymptotic arguments are required within this proof, we omit the index $n$ throughout. We follow the corresponding proof for the ADDIS-Graph in \citet{fischer2023adaptive}. Let $(\alpha_i)_{i\in\N}$ be the test levels of the continuous Adaptive-Graph as defined  in \eqref{def_test_levels_dependent_graphgeneral_case}. We need to show that for every $i\in\N$ it holds that 
\[\sum_{j=1}^i\frac{\alpha_{j}}{1-\lambda_j}\xi_{j}\leq\alpha.\]
We fix an $i\in\N$ and define $U_j=1-\xi_j$ for $j\in\{1,\ldots ,i\}$. Then the claim is equivalent to  
\begin{equation}\label{eqn_appendix_graph_proof}
	F_i(U_{1:i})=\sum_{j=1}^{i}\left\{\alpha\gamma_j + \sum_{k=1}^{j-1}g_{k,j}U_k\alpha_k (U_{1:k-1})(1-\lambda_k)^{-1}\right\} (1-U_j)\leq\alpha
\end{equation}
for any $U_{1:i}=(U_1,\ldots ,U_i)^T\in [0,1]^{i}$.
First, consider the special case $U_{1:i}^0=(0,\ldots ,0)^T\in [0,1]^{i}$, which obviously satisfies \eqref{eqn_appendix_graph_proof}, since $F_i(U_{1:i}^0)=\sum_{j=1}^{i}\alpha\gamma_j\leq\alpha$. Now, let $U_{1:i}=(U_1,\ldots ,U_i)^T\in [0,1]^{i}$ be arbitrary but fixed. The plan is to show that $F_i(U_{1:i})\leq F_i(U_{1:i}^0)$, which would conclude the proof. Without loss of generality, we assume that $U_j> 0$ for at least one $j\in\{1,\ldots ,i\}$ as we have already dealt with the case of all entries being zero. Furthermore, define $l=\max\{j\in\{1,\ldots ,i\}:U_j>0\}$ and consider $U_{1:i}^l=(U_1^l,\ldots ,U_i^l)^T$, where $U_j^l=U_j$ if $j\neq l$ and $U_l^l=0$. For ease of notation we write from now on $\alpha_j$ and $\alpha_j^l$ instead of $\alpha_j(U_{1:j-1})$ and $\alpha_j(U_{1:j-1}^l)$, respectively. By the definition of $l$ and $U_{1:i}^l$, for every $j\leq i$ we have that $U_{j}^l=0$ for $j\geq l$, $U_{j}=0$ for $j\geq l+1$ and $\alpha_j^l=\alpha_j$ for $j\leq l$.
Thus, we obtain
\begin{align*}
	F_i(U_{1:i}^l)-F_i(U_{1:i})&=\sum_{j=1}^{i}\alpha\gamma_j (U_j-U_j^l)+\sum_{j=1}^{i}\left\{\sum_{k=1}^{j-1}g_{k,j}U_k^l\alpha_k^l(1-\lambda_k)^{-1}\right\}(1-U_j^l)\\
	& ~~~~ -\sum_{j=1}^{i}\left\{\sum_{k=1}^{j-1}g_{k,j}U_k\alpha_k(1-\lambda_k)^{-1}\right\}(1-U_j)\\
	&=\alpha\gamma_lU_l+\sum_{k=1}^{l-1}g_{k,l}U_k^l\alpha_k^l(1-\lambda_k)^{-1} - (1-U_l)\sum_{k=1}^{l-1}g_{k,l}U_k\alpha_k(1-\lambda_k)^{-1} \\
	& ~~~~ +\sum_{j=l+1}^{i}\left\{\sum_{k=1}^{j-1}g_{k,j}U_k^l\alpha_k^l(1-\lambda_k)^{-1}\right\}-\sum_{j=l+1}^{i}\left\{\sum_{k=1}^{j-1}g_{k,j}U_k\alpha_k(1-\lambda_k)^{-1}\right\}\\
	&=\alpha\gamma_lU_l + U_l\sum_{k=1}^{l-1}g_{k,l}U_k\alpha_k(1-\lambda_k)^{-1} - \sum_{j=l+1}^{i}g_{l,j}U_l\alpha_l(1-\lambda_l)^{-1}\\
	&\geq U_l\left\{\alpha\gamma_l + \sum_{k=1}^{l-1}g_{k,l}U_k\alpha_k(1-\lambda_k)^{-1} - \alpha_l(1-\lambda_l)^{-1}\right\}\stackrel{\eqref{def_test_levels_dependent_graphgeneral_case}}{=}0,
\end{align*}
where the inequality holds, because $\sum_{j>l}g_{l,j}\leq 1$ and $U_l,\alpha_l,(1-\lambda_l)^{-1}\geq 0$. So, we have $F_i(U_{1:i})\leq F_i(U_{1:i}^l)$. Iterating this process, it follows that $F_i(U_{1:i})\leq F_i(U_{1:i}^0)$, which completes the proof.
\end{proof}

\begin{proof}[Proof that the continuous Adaptive-Graph is a generalization of the procedure in Example \ref{example_procedure_uniform}]
 Let $\lambda_i,\Pi\in (0,1)$, $\gamma_i=\Pi (1-\Pi)^{i-1}, i\in\N$, and $g_{j,i}=\gamma_{i-j}$ for all $j\in\N,i>j$.  Further, let $\alpha_{i,n}^{\text{Graph}}$ and $\tilde{\alpha}_{i,n}$ denote the test levels of the continuous Adaptive-Graph \eqref{def_test_levels_dependent_graphgeneral_case} and the geometric spending algorithm \eqref{eqn_test_levels_canonic_proc}, respectively. We give a proof by induction. By definition it holds that 
\[\alpha_{1,n}^{\text{Graph}}=(1-\lambda_1)\gamma_1\alpha=(1-\lambda_1)\Pi\alpha=\tilde{\alpha}_{1,n}.\]

Now assume that $\alpha_{i,n}^{\text{Graph}}=\tilde{\alpha}_{i,n}$ holds for an $i\in\N$. Then,
\begin{align*}
	\alpha_{i+1,n}^{\text{Graph}}&=(1-\lambda_{i+1})\left\{\gamma_{i+1}\alpha + \sum_{j=1}^{i}\gamma_{i-j+1}(1-\xi_{j,n})\frac{\alpha_{j,n}^{\text{Graph}}}{1-\lambda_j}\right\}\\
	&= (1-\lambda_{i+1}) \left\{(1-\Pi)\gamma_i\alpha +  \Pi(1-\xi_{i,n})\frac{\alpha_{i,n}^{\text{Graph}}}{1-\lambda_i} + (1-\Pi)\sum_{j=1}^{i-1}\gamma_{i-j}(1-\xi_{j,n})\frac{\alpha_{j,n}^{\text{Graph}}}{1-\lambda_j}  \right\}\\
	&= (1-\lambda_{i+1}) \left\{ \Pi(1-\xi_{i,n})\frac{\alpha_{i,n}^{\text{Graph}}}{1-\lambda_i} + (1-\Pi)\frac{\alpha_{i,n}^{\text{Graph}}}{1-\lambda_i}  \right\}\\
	&= (1-\lambda_{i+1})\frac{\alpha_{i,n}^{\text{Graph}}}{1-\lambda_i}(1-\Pi\xi_{i,n})\\
	&= (1-\lambda_{i+1})\frac{\tilde{\alpha}_{i,n}}{1-\lambda_i}(1-\Pi\xi_{i,n})\\
	&\stackrel{\eqref{def_testlevel_uniform_recursion}}{=}\tilde{\alpha}_{i+1,n}.
\end{align*}
\end{proof}

\begin{proof}[Proof of Lemma~\ref{lemma_continuous__spending_test_levels_conservative}]
Let $j_1<\cdots<j_k$ denote the indices corresponding to the true null hypotheses among the first $N$ hypotheses and let $0<\varepsilon <1-\lambda$. Due to \eqref{eqn_assumption_xi_consistent} there exists an $n_0\in\N$ such that for all $n\geq n_0$ we have $\PB(1-\lambda-\xi_{i,n}\geq\varepsilon)\leq \varepsilon/k$ for all $i\in\mathcal{H}^0=\mathcal{H}^0(N)$.
For the test levels defined in \eqref{def_test_level_continuous_spending} it holds that
\begin{align*}
	\E\left(\sum_{i\in\{1,\ldots ,k\}}\alpha_{j_i,n}\right)&\leq\E\left\{\frac{\alpha}{s}  \sum_{i\in\{1,\ldots ,k\}}(1-\lambda)f \left(1 + \sum_{l<i}\xi_{j_l,n}\right)\right\}\\
	&=\frac{\alpha}{s} \E\left\{(1-\lambda)f (1)+ \sum_{i\in\{2,\ldots ,k\}}\1\{\xi_{j_{i-1},n}>1-\lambda-\varepsilon\}(1-\lambda)f \left(1 + \sum_{l<i}\xi_{j_l,n}\right)\right\}\\
	&\phantom{{}=1} + \frac{\alpha}{s} \E\left\{ \sum_{i\in\{2,\ldots ,k\}}\1\{\xi_{j_{i-1},n}\leq 1-\lambda-\varepsilon\}(1-\lambda)f \left(1 + \sum_{l<i}\xi_{j_l,n}\right)\right\}\\
	&\leq\frac{\alpha}{s} \E\left\{(1-\lambda)f(1)+\sum_{i\in\{2,\ldots ,k\}}\1\{\xi_{j_{i-1},n}>1-\lambda-\varepsilon\}\frac{1-\lambda}{\xi_{j_{i-1},n}}\int_{1  + \sum_{l<i-1}\xi_{j_l,n}}^{1 + \sum_{l<i}\xi_{j_l,n}}f (x)dx\right\}\\
	&\phantom{{}=1} + \frac{\alpha}{s} (1-\lambda) f(1)\sum_{i\in\{2,\ldots ,k\}}\PB(\xi_{j_{i-1},n}\leq 1-\lambda-\varepsilon)\\
	&\leq\frac{\alpha}{s}  \E\left\{(1-\lambda)f(1)+\sum_{i\in\{2,\ldots ,k\}}\frac{1-\lambda}{1-\lambda-\varepsilon}\int_{1  + \sum_{l<i-1}\xi_{j_l,n}}^{1 + \sum_{l<i}\xi_{j_l,n}}f (x)dx\right\}\\
	&\phantom{{}=1} + \alpha \sum_{i\in\{2,\ldots ,k\}}\frac{\varepsilon}{k}\\
	&\leq\frac{\alpha}{s} \E\left\{\frac{1-\lambda}{1-\lambda-\varepsilon} (1-\lambda)f(1)+\sum_{i\in\{2,\ldots ,k\}}\frac{1-\lambda}{1-\lambda-\varepsilon}\int_{1  + \sum_{l<i-1}\xi_{j_l,n}}^{1 + \sum_{l<i}\xi_{j_l,n}}f (x)dx\right\} + \varepsilon\\
	&\leq \frac{\alpha(1-\lambda)}{s(1-\lambda-\varepsilon)} \E\left\{(1-\lambda)f(1)+\int_{1}^{\infty}f (x)dx\right\} + \varepsilon\\
	&=\alpha \frac{1-\lambda}{1-\lambda-\varepsilon}+\varepsilon.
\end{align*}

Here, in the first and second inequality it is used, that $f$ is non-increasing. 
As $\varepsilon$ was arbitrary, we conclude $\lim_{n\rightarrow\infty}\E\left(\sum_{i\in\{1,\ldots ,k\}}\alpha_{j_i,n}\right)\leq\alpha$. Since $\alpha_{i,n}\to \alpha_i$ in probability, we also have that $\sum_{i\in\mathcal{H}^0}\alpha_{i,n}\to \sum_{i\in\mathcal{H}^0}\alpha_i$ in probability. The dominated convergence theorem for convergence in probability (cf. \cite{bogachev2007measure}, Theorem 2.8.5) plus the fact that $\sum_{i\in\mathcal{H}^0}\alpha_{i,n}\leq k$ then yield \[\alpha\geq\E\left(\sum_{i\in\mathcal{H}^0}\alpha_i\right)=\sum_{i\in\mathcal{H}^0}\alpha_i.\]
This concludes the proof.
\end{proof}

\emph{Computation of the constant for the interpolation procedure \eqref{def_test_level_interpolation}}.
We compute
\begin{align*}
	\int_1^\infty f_\gamma (x)dx &= \sum_{k=1}^\infty \int_k^{k+1}f_\gamma (x)dx\\
	&=\sum_{k=1}^\infty \int_k^{k+1} \left\{\gamma_k + (x-k)(\gamma_{k+1}-\gamma_k)\right\}dx\\
	&=\sum_{k=1}^\infty \left\{\gamma_k + \int_0^{1} x(\gamma_{k+1}-\gamma_k)dx\right\}\\
	&= \frac{1}{2}\left(\sum_{k=1}^\infty \gamma_k + \sum_{k=2}^\infty \gamma_k\right)\\
	&=1-\frac{\gamma_1}{2}.
\end{align*}
Since $f_\gamma(1)=\gamma_1$, we have $s=(1-\lambda)f_\gamma(1)+ \int_1^\infty f_\gamma(x)dx=1+\gamma_1(1/2-\lambda)$.

\emph{Derivation of the closed procedures \eqref{def_test_levels_dependent_graph_closed} and \eqref{def_test_level_continuous_spending_closed}}.
Based on the test levels \eqref{def_test_levels_dependent_graphgeneral_case}, for every $I\subseteq \N$ we consider the intersection test that rejects the intersection hypothesis $H_I=\bigcap_{i\in I} H_i$ if and only if for at least one $i\in I$ it holds that
\[P_{i,n}\leq\alpha_{i,n}^{I}=(1-\lambda_i)\left\{\alpha\gamma_i + \sum_{j\in I,j<i}g_{j,i}(1-\xi_{j,n})\frac{\alpha_{j,n}^{I}}{1-\lambda_j} + \sum_{j\notin I,j<i}g_{j,i}\frac{\alpha_{j,n}^{(I\cup \{j\})}}{1-\lambda_j} \right\}.\]
Note that  $\alpha_{1,n}^{I}=(1-\lambda_1) \alpha\gamma_1$ if $1\in I$ and all further test levels are defined recursively. 
It is easy to see, that this defines a predictable and consonant intersection test (with the terminology as in \citet{fischer2022online}), since the test levels depend only on previous indices and for any $i\in I$ it holds that $\alpha_{i,n}^{I}\leq \alpha_{i,n}^J$ for all $J\subseteq I$ with $i\in J$. Moreover, it is an (asymptotical) level $\alpha$ test, as it fulfills \eqref{def_general_dependent_algorithm} (with the summation only taking place over the index set $I$), because the original procedure \eqref{def_test_levels_dependent_graphgeneral_case} does and the case $j\notin I$ can be interpreted as $\xi_{j,n}=0$. Now, we recursively define $I_1=\{1\}$ and $I_i=\{j\in\N:j<i,P_{j,n}>\alpha_{j,n}^{I_j}\}\cup \{i\}$ for $i\geq 2$. The short-cut procedure according to \citet{fischer2022online} is then given by
\begin{align*}\alpha_{i,n}=\alpha_{i,n}^{I_i}=(1-\lambda_i)\left(\alpha\gamma_i + \sum_{j<i}g_{j,i}\max\{1-\xi_{j,n},R_{j,n}\}\frac{\alpha_{j,n}}{1-\lambda_j}  \right).
\end{align*}

Analogous, the test levels for the intersection tests based on the procedure \eqref{def_test_level_continuous_spending} are defined as 
\[\alpha_{i,n}^{I}=\alpha \frac{1 - \lambda}{s} f \left(1 + \sum_{j\in I,j<i}\xi_{j,n}\right).\]
Again, it is easy to verify that these intersection tests are predictable and consonant. Furthermore, they are asymptotical level $\alpha$ tests, as can be seen by slightly modifying the proof of Lemma \ref{lemma_continuous__spending_test_levels_conservative}. With $I_i$ as above, the short-cut procedure is defined as 
\begin{align*}\alpha_{i,n}=\alpha_{i,n}^{I_i}=\alpha \frac{1 - \lambda}{s} f \left\{1 + \sum_{j<i}\xi_{j,n} (1-R_{j,n})\right\}.
\end{align*}

\emph{Obtaining consistent weights in a 2-sample design}.
Assume that the observed data $\mathcal{X}$ consists of independent random variables $X_{1},\ldots ,X_{n_1}\sim\mathcal{N}(\mu^X,1),Y_{1},\ldots ,Y_{n_2}\sim\mathcal{N}(\mu^Y,1)$ and we want to test the hypothesis $H:\mu^X-\mu^Y \leq 0$ versus the alternative $H':\mu^X-\mu^Y>0$. Here, $n_1=n_1(n)$ and $n_2=n_2(n)$ are the sample sizes for the respective groups, which depend implicitly on the overall asymptotic parameter $n$. The test statistic and the corresponding one-sided $p$-value are \mbox{$Z_{n}=(1\slash n_1 + 1\slash n_2)^{-1/2}(\overline{X}_{n_1}-\overline{Y}_{n_2})$} and $P_{n}=1-\Phi (Z_{n})$, respectively. 
Similar to Example \ref{example_low_intensity} consider the parametric bootstrap variables $X_{1}^*,\ldots ,X_{m(n_1)}^*,Y_{1}^*,\ldots ,Y_{m(n_2)}^*$, which are independently drawn from the distributions
$\mathcal{N}(\overline{X}_{n_1},1)$ \text{ and } $\mathcal{N}(\overline{Y}_{n_2},1)$,
respectively.
As before, $m:\N\to\N$ is a non-decreasing map such that $m(n)\to\infty$ and $m(n)/n\to 0$ as $n\to\infty$. 
Furthermore, we construct the (low intensity) bootstrap test statistic  
\[Z_{m(n)}^*=\{1\slash m(n_1) + 1\slash m(n_2)\}^{-1/2}\left\{m(n_1)^{-1}\sum_{j=1}^{m(n_1)}X_{j}^* - m(n_2)^{-1}\sum_{j=1}^{m(n_2)}Y_{j}^*\right\},\]
whose conditional distribution is a normal distribution with variance 1 and mean 
\[\frac{(1\slash n_1 + 1\slash n_2)^{1/2}}{\{1\slash m(n_1) + 1\slash m(n_2)\}^{1/2}}Z_{n}.\]
Moreover, we define the low intensity bootstrap version of the $p$-value $P_{n}$ by $P_{m(n)}^*=1-\Phi \{Z^*_{m(n)}\}$.
For any $\lambda\in (0,1)$, we can now define the consistent weights by $\xi_n=\PB\{P_{m(n)}^*>\lambda|\mathcal{X}\}$. The consistency condition \eqref{eqn_assumption_xi_consistent} can be shown completely analogous to Example \ref{example_low_intensity} if $n_1=n_2$. In general, it is necessary that the groups are at least asymptotically balanced, i.e. $n_1/n_2\to 1$ as $n\to\infty$. 

\subsection*{Further simulation results}

\begin{figure}[H]
\centering
\includegraphics[scale=0.75]{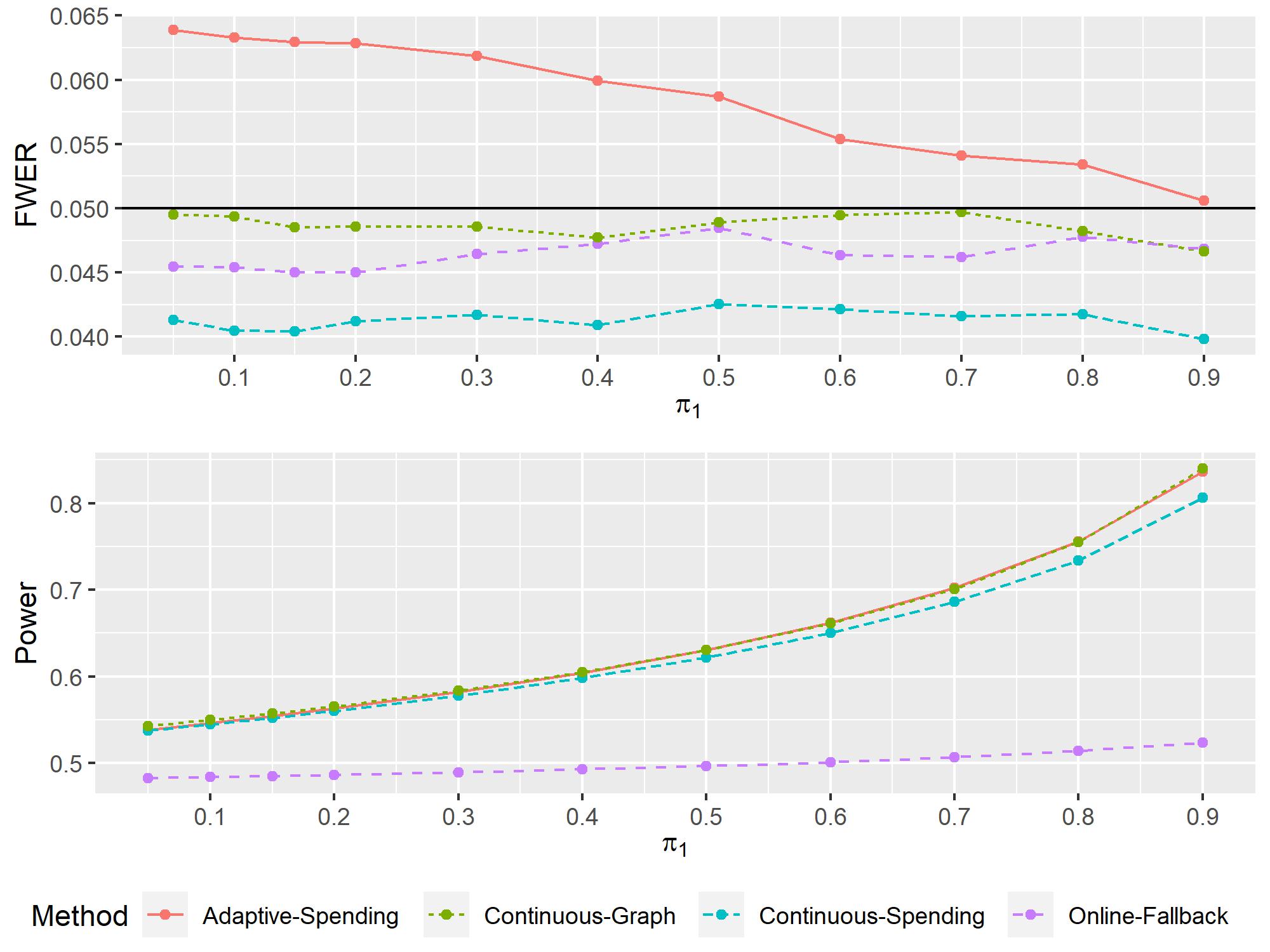}
\caption{Estimated FWER and power in an autocorrelation set-up with $\rho=0.5$ and target level $\alpha=0.05$ plotted over the proportion of false null hypotheses $\pi_1$.}
\end{figure}

\begin{figure}[H]
\centering
\includegraphics[scale=0.75]{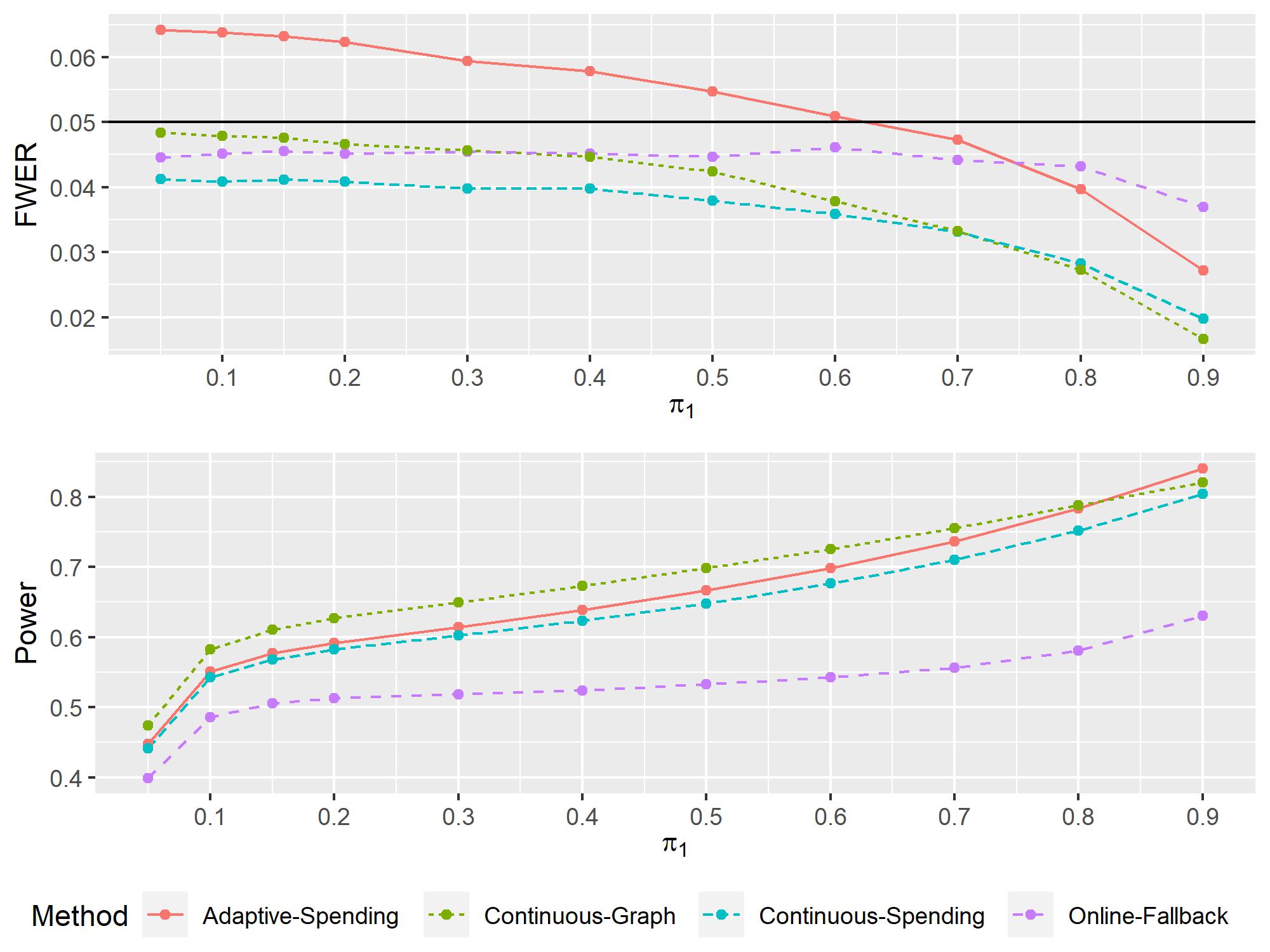}
\caption{Estimated FWER and power in a platform trial scenario with $N=30$ and target level $\alpha=0.05$ plotted over the proportion of false null hypotheses $\pi_1$.}
\end{figure}

\begin{figure}[H]
\centering
\includegraphics[scale=0.75]{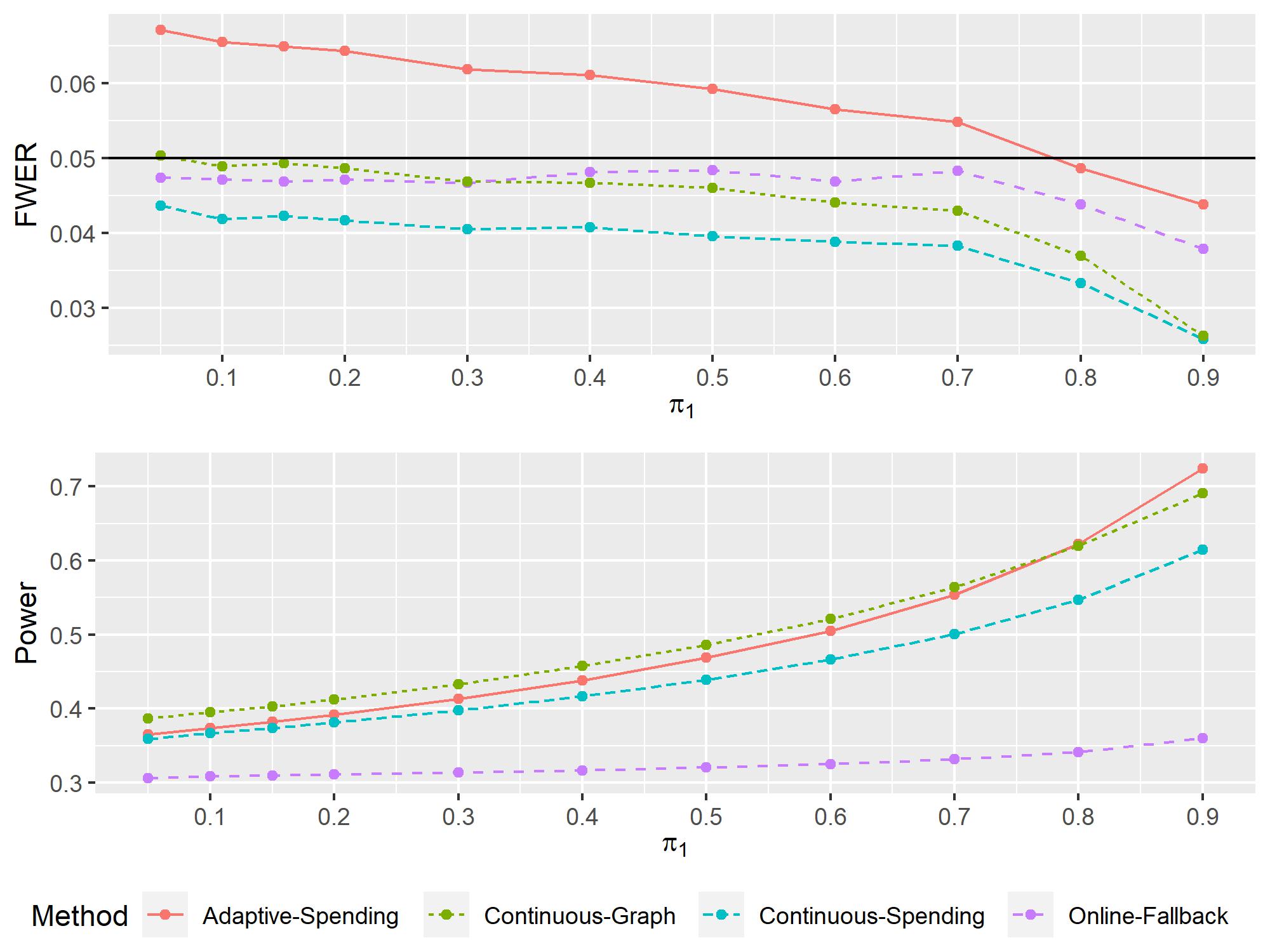}
\caption{Estimated FWER and power in a platform trial scenario with $N=100$ and target level $\alpha=0.05$ plotted over the proportion of false null hypotheses $\pi_1$.}
\end{figure}

\begin{figure}[H]
\centering
\includegraphics[scale=0.8]{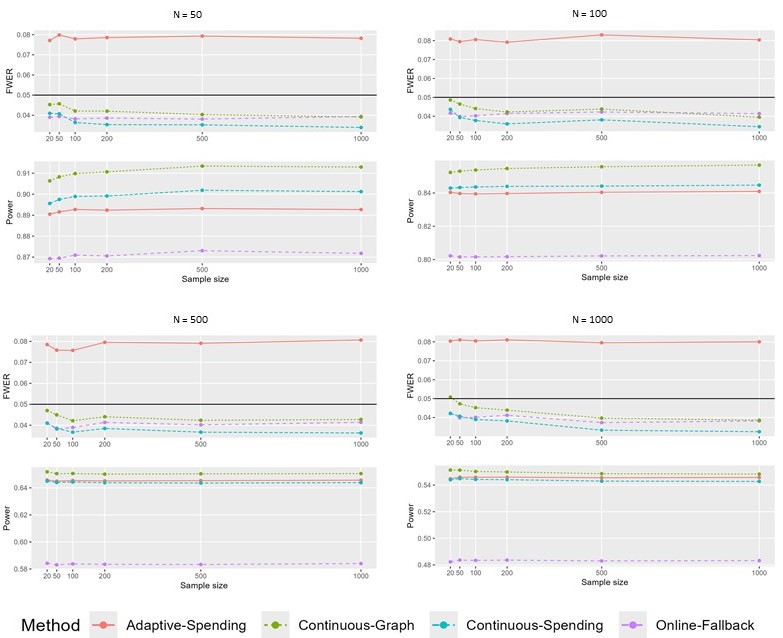}
\caption{Estimated FWER and power in an autocorrelation set-up with $\rho=0.8$, $\pi_1=0.1$ and target level $\alpha=0.05$ for $N\in\{50,100,500,1000\}$ hypotheses plotted over the sample size.}
\end{figure}

\section*{Funding}
Lasse Fischer acknowledges funding by the Deutsche Forschungsgemeinschaft (DFG, German Research Foundation) – Project number 281474342/GRK2224/2.

\section*{Supplementary Material}
The R source code used for the simulations is available under \href{https://github.com/vijankovic/Online-Procedures-Dependency}{https://github.com/vijankovic/Online-Procedures-Dependency}.

\bibliographystyle{plainnat} 

\bibliography{references}

\end{document}